\documentclass[12pt]{article}
\usepackage{amsmath, amsthm}
\usepackage{graphicx,psfrag,epsf}
\usepackage{enumerate}

\newcommand{\blind}{0}

\newtheorem{theorem}{\bf Theorem}
\newtheorem{proposition}{\bf Proposition}

\numberwithin{equation}{section}

\begin{document}

\def\spacingset#1{\renewcommand{\baselinestretch}%
{#1}\small\normalsize} \spacingset{1}

%%%%%%%%%%%%%%%%%%%%%%%%%%%%%%%%%%%%%%%%%%%%%%%%%%%%%%%%%%%%%%%%%%%%%%%%%%%%%%

\if0\blind
{\title{Improved EM for Mixture Proportions with Applications to Nonparametric ML Estimation for Censored Data}
  \author{Yaming Yu\\
    Department of Statistics, University of California, Irvine}
  \maketitle
} \fi

\if1\blind
{
  \bigskip
  \bigskip
  \bigskip
  \begin{center}
    {\LARGE\bf Improved EM for Mixture Proportions with Applications to Nonparametric ML Estimation for Censored Data}
\end{center}
  \medskip
} \fi

\bigskip
\begin{abstract}
Improved EM strategies, based on the idea of efficient data augmentation (Meng and van Dyk 1997, 1998), are presented for ML estimation of mixture proportions.  The resulting algorithms inherit the simplicity, ease of implementation, and monotonic convergence properties of EM, but have considerably improved speed.  Because conventional EM tends to be slow when there exists a large overlap between the mixture components, we can improve the speed without sacrificing the simplicity or stability, if we can reformulate the problem so as to reduce the amount of overlap.  We propose simple ``squeezing'' strategies for that purpose.  Moreover, for high-dimensional problems, such as computing the nonparametric MLE of the distribution function with censored data, a natural and effective remedy for conventional EM is to add exchange steps (based on improved EM) between adjacent mixture components, where the overlap is most severe.  Theoretical considerations show that the resulting EM-type algorithms, when carefully implemented, are globally convergent.  Simulated and real data examples show dramatic improvement in speed in realistic situations. 
\end{abstract}

\noindent%
{\it Keywords:} AECM; cocktail algorithm; data augmentation; doubly censored data; EM; global convergence; NPMLE; nonparametric mixtures; squeezing; vertex exchange method. 

\spacingset{1.45}

\section{Introduction}
Several statistical problems give rise to a likelihood function formally equivalent to that of a finite mixture model with known component densities.  One example is maximum likelihood (ML) estimation in a saturated multinomial model with ignorable missing data where some units are partially classified (Dempster et al.\ 1977).  Another example is nonparametric ML estimation (NPMLE) of a mixing distribution (Lindsay 1983) when this distribution is assumed to be supported on a finite grid.  Closely related is the NPMLE problem for the distribution function for censored data (Groeneboom and Wellner 1992; B\"{o}hning et al.\ 1996).  The EM algorithm (Dempster et al.\ 1977; Meng and van Dyk 1997) is among the simplest and best known methods for ML computation in such mixture-like problems; Turnbull (1976) used it on censored data before Dempster et al.\ (1977) laid down the general framework.  The potential slow convergence of EM in general, and for NPMLE computation in particular, is also well documented.  Other methods of computing the NPMLE include the iterative convex minorant (ICM) algorithm (Aragon and Eberly 1992; Jongbloed 1998), the vertex exchange method (VEM; B\"{o}hning et al.\ 1996), and constrained Newton methods (Wang 2008).  EM or EM-like algorithms are also widely used for related problems such as optimal experimental design (Silvey et al.\ 1978; Yu 2010a), Poisson image reconstruction using positron emission tomography (Vardi et al.\ 1985), and channel capacity calculations in Shannon theory (Arimoto 1972; Blahut 1972; Csisz\'{a}r and Tusn\'{a}dy 1984). 

This paper is concerned with improved EM strategies for maximizing a finite mixture log-likelihood with known component densities.  Possible extensions to more general problems are mentioned in Section~5.  Our main motivation is fast computation of the NPMLE for censored data.  The NPMLE problem is challenging partly because of the high dimension (there are many mixture components), and the heavy overlap between components, which slows down conventional EM.  Our goal is to design algorithms that improve the speed of EM, but preserve its simplicity, ease of implementation, and  monotonic convergence properties.  First, we introduce ``squeezing'' strategies that reformulate the problem so as to reduce the overlap between component densities.  Such  squeezing strategies capitalize on the idea of efficient data augmentation and are inspired by Fessler and Hero (1994).  The resulting EM algorithms converge faster because they correspond to augmented data that are less informative.  Secondly, we observe that although ``squeezing'' may not always be effective for the entire collection of mixture components, we can apply it to sub-collections that overlap most severely.  Adding such EM-based conditional maximization steps (nearest neighbor exchanges) can improve the speed dramatically.  Overall, our algorithms fit in the broad spectrum of alternating-expectation-conditional-maximization (AECM) schemes (Meng and van Dyk 1997, 1998).  The simplicity and effectiveness of these algorithms testify to the advantage of working within the general EM framework (Dempster et al.\ 1977; Wu 1983; Meng and Rubin 1993; Liu and Rubin 1994; Meng and van Dyk 1997; Liu et al.\ 1998).  Also relevant is the work of Pilla and Lindsay (2001), who focus on the nonparametric mixture problem and propose pairing nearby components and rotating the pairs for fast ML computation. 

In Section~2, we introduce the squeezing strategies to improve EM for maximizing a mixture log-likelihood.  Intuitively, squeezing yields an equivalent problem where the mixture components have less overlap, and its effect on the speed of EM is explained in terms of efficient data augmentation.  Section~3 argues that squeezing strategies can be effectively implemented to sub-collections of the mixture components.  This leads to a ``cocktail algorithm'' with several different moves that complement each other.  A global convergence theorem for the cocktail algorithm is proved.  A real-data example is included as an illustration.  Section~4 focuses on the NPMLE problem for censored data, and demonstrates the effectiveness of our new algorithms through simulation.  Section~5 concludes with a discussion on possible extensions (to the bivariate interval censoring problem, for example).  Efficient implementations of our EM-type algorithms, which take advantage of the sparsity features of the NPMLE problem for censored data, are collected in the appendix. 

\section{EM Algorithms for Mixture Proportions}
Suppose $n$ observations $y=(y_1,\ldots, y_n)$ are taken from a mixture of $m$ known densities with unknown proportions $p_1, \ldots, p_m$.  Writing $f_{ij}$ as the $j$th component density evaluated at $y_i$, we can express the log-likelihood for $\mathbf{p}=(p_1,\ldots, p_m)$ as 
\begin{equation}
\label{llike}
l(\mathbf{p})=\sum_{i=1}^n \log\left(\sum_{j=1}^m f_{ij} p_j\right).
\end{equation}
We seek to maximize (\ref{llike}) over $\mathbf{p}\in \Theta$ where 
$$\Theta=\left\{\mathbf{p}:\ \sum_{j=1}^m p_j=1,\ p_j\geq 0;\ l(\mathbf{p})>-\infty\right\}.$$ 

\subsection{Conventional EM}
Conventional EM introduces latent indicators $I_{ij}$ such that $I_{ij}=1$ if the $i$th observation is from component $j$, and $I_{ij}=0$ otherwise.  At iteration $t$, when the current estimate of $\mathbf{p}$ is $\mathbf{p}^{(t)}=\left(p_1^{(t)}, \ldots, p_m^{(t)}\right)$, the E-step simply computes the 
conditional expectation of $I_{ij}$ given observed data and $\mathbf{p}^{(t)}$:
$$E\left(I_{ij}\left|y, \mathbf{p}^{(t)}\right)\right.= \frac{f_{ij}p_j^{(t)}}{\sum_{k=1}^m f_{ik} p_k^{(t)}}.$$
The M-step then sets $p_j^{(t+1)}=\sum_{i=1}^n E\left(I_{ij}|y, \mathbf{p}^{(t)}\right)/n.$  Overall each iteration can be written as
\begin{equation}
p_j^{(t+1)}=\frac{1}{n}\sum_{i=1}^n \left(\frac{f_{ij}}{\sum_{k=1}^m f_{ik} p_k^{(t)}}\right) p_j^{(t)},\quad 
j=1,\ldots, m. 
\label{em}
\end{equation} 
The EM algorithm maintains monotone increase in the log-likelihood, i.e., $l(\mathbf{p}^{(t+1)})\geq l(\mathbf{p}^{(t)})$.  
Moreover, when started from the interior of the parameter space, i.e., when $p^{(0)}_j>0$ for all $1\leq j\leq m$, EM is guaranteed to converge to $\mathbf{\hat{p}}$, the MLE (Csisz\'{a}r and Tusn\'{a}dy 1984).  Convergence is potentially very 
slow, however, when there exists heavy overlap among the mixture components. 

\subsection{Squeezing Strategy I}
To improve conventional EM, let us introduce an auxiliary vector $g=(g_1, \ldots, g_n)$ and write the objective function (\ref{llike}) as
\begin{align}
\label{llike0}
l(\mathbf{p})&=n\log 2+\sum_{i=1}^n \log\left(\sum_{j=1}^m g_i p_j/2+\sum_{j=1}^m (f_{ij}-g_i) p_j/2 \right)\\
\label{llike1}
&=n\log 2+\sum_{i=1}^n \log\left(g_i/2+\sum_{j=1}^m (f_{ij}-g_i) p_j/2 \right).
\end{align}
We require $g_i\geq 0$ and $f_{ij}-g_i\geq 0$ for all $i$, i.e.,
\begin{equation}
\label{cond1}
0\leq g_i\leq \min_j f_{ij}.
\end{equation}

Conventional EM (\ref{em}) can also be derived from (\ref{llike0}), viewing it as a mixture log-likelihood with $2m$ components with proportions $p_j/2,\ j=1,\ldots, m$, each appearing twice.  Specifically, under this new formulation, we let the density of component $j$ ($j=1, \ldots, 2m$) evaluated at observation $i$ be 
 $$\tilde{f}_{ij}=\begin{cases} f_{ij}-g_i, & 1\leq j\leq m,\\ g_i, & m+1\leq j\leq 2m.\end{cases}$$ 
Let $\tilde{I}_{ij}$ be the latent indicator of whether the $i$th observation is from component $j,\ j=1,\ldots, 2m$.  The $j$th component has proportion $p_j/2$ if $1\leq j\leq m$, and $p_{j-m}/2$ if $m+1\leq j\leq 2m$.  Then the E-step becomes 
\begin{equation}
\label{e0}
E\left(\left.\tilde{I}_{ij}\right|y, \mathbf{p}^{(t)}\right)= \frac{\tilde{f}_{ij}}{g_i+\sum_{k=1}^m \tilde{f}_{ik} p_k^{(t)}}\times\begin{cases} p_j^{(t)}, & 1\leq j\leq m,\\ p_{j-m}^{(t)}, & m+1\leq j\leq 2m, \end{cases}
\end{equation}
and the M-step becomes 
\begin{equation}
\label{m0}
p_j^{(t+1)}\propto \sum_{i=1}^n E\left(\tilde{I}_{ij}+\left.\tilde{I}_{i,j+m}\right|y, \mathbf{p}^{(t)}\right),\quad j=1,\ldots, m.
\end{equation}
Routine algebra reveals that the resulting EM iteration is the same as (\ref{em}). 

We can also apply EM to maximize (\ref{llike1}), viewing it as a problem with $m+1$ mixture components, one of which has proportion $1/2$.  Equivalently, in the above derivation of (\ref{e0}) and (\ref{m0}), instead of $\tilde{I}_{ij},\ 1\leq i\leq n,\ 1\leq j\leq 2m$, let us treat 
$$\tilde{I}_{ij},\ 1\leq i\leq n,\ 1\leq j\leq m,\quad {\rm and}\quad \tilde{I}_{i0}\equiv \sum_{j=m+1}^{2m} \tilde{I}_{ij},\ 1\leq i\leq n,$$ 
as the set of latent indicators.  The latter, being a collapsed version of the former, contains less information about $\mathbf{p}$.  Note that $\tilde{I}_{i0}$ is the indicator of a mixture component with proportion $1/2$, whose density evaluated at observation $i$ is $g_i$.  The E-step proceeds to calculate the conditional expectation of $\tilde{I}_{ij},\ j=0, 1, \ldots, m$, resulting in the same formula as (\ref{e0}) for $j=1, \ldots, m$.  The M-step becomes
\begin{equation}
\label{m1}
p_j^{(t+1)}\propto \sum_{i=1}^n E\left(\left.\tilde{I}_{ij}\right|y, \mathbf{p}^{(t)}\right),\quad j=1,\ldots, m.
\end{equation}
The conditional expectation of $\tilde{I}_{i0}$ does not appear in (\ref{m1}).  Combining (\ref{e0}) (for $j=1,\ldots, m$) with (\ref{m1}), we obtain an EM iteration as 
\begin{equation}
p_j^{(t+1)} =\left(n-\sum_{i=1}^n \frac{g_i}{\sum_{k=1}^m f_{ik} p_k^{(t)}}\right)^{-1} \sum_{i=1}^n \left(\frac{f_{ij}-g_i}{\sum_{k=1}^m f_{ik} p_k^{(t)}}\right) p_j^{(t)},\quad j=1,\ldots, m.
\label{em1}
\end{equation} 
The iteration (\ref{em}) corresponds to (\ref{em1}) with $g\equiv 0$. 

Because (\ref{em1}) is derived in the EM framework, it inherits nearly all the desirable properties of (\ref{em}).  For example, each iteration of (\ref{em1}) increases the log-likelihood (\ref{llike1}).  Furthermore, because the convergence rate of EM is determined by the fraction of missing information, we know that (\ref{em1}) converges faster than (\ref{em}) because it is based on a reduced set of latent variables.  This is an example of efficient data augmentation (Meng and van Dyk 1997): we speed up EM by augmenting less.  This strategy of improving (\ref{em}) is called a squeezing strategy because the key equivalent formula (\ref{llike1}) is obtained by subtracting (``squeezing out'') a nonnegative vector $g$ from each component density. 

A slight extension of the above discussion shows that the convergence rate of (\ref{em1}) is monotonic in the squeezing parameter $g$.  Because of the restriction (\ref{cond1}), the optimal $g$ is therefore the upper bound $g_i=\min_j f_{ij}$, that is, we perform as much squeezing as allowed. 
This optimal $g$ may be viewed as the overlap among all component densities.  If this overlap is small, i.e., $g$ is close to a vector of zeros, then (\ref{em1}) is not very different from (\ref{em}).  Hence we may expect significant speedup only if the overlap is large enough. 

This squeezing strategy (as well as the strategy of Section~2.3) is inspired by the work of Fessler and Hero (1994) on efficient EM for Poisson image problems.  Similar strategies also work for the Arimoto-Blahut algorithm for calculating the Shannon capacity of a discrete memoryless channel (see Yu 2010b). 

\subsection{Squeezing Strategy II}
The squeezing strategy of Section~2.2 can be improved by further manipulation of the log-likelihood.  Let us introduce another auxiliary vector $\beta=(\beta_1,\ldots, \beta_m)$, and rewrite (\ref{llike1}) as
\begin{align}
\label{llike00}
l(\mathbf{p})&=n\log (2+\beta_+) + \sum_{i=1}^n \log \frac{g_i-\sum_{j=1}^m \tilde{f}_{ij} \beta_j + \sum_{j=1}^m \tilde{f}_{ij} p_j +\sum_{j=1}^m \tilde{f}_{ij} \beta_j}{2+\beta_+}\\
\label{llike2}
&=n\log (2+\beta_+) + \sum_{i=1}^n \log\frac{g_i-\sum_{j=1}^m \tilde{f}_{ij} \beta_j +\sum_{j=1}^m \tilde{f}_{ij} (p_j+\beta_j)}{2+\beta_+},
\end{align}
where $\tilde{f}_{ij}=f_{ij}-g_i$ as before, and $\beta_+=\sum_{j=1}^m \beta_j$.  It is required that 
\begin{equation}
\label{cond2}
\beta_j\geq 0,\ 1\leq j\leq m;\quad g_i-\sum_{j=1}^m \tilde{f}_{ij}\beta_j \geq 0,\ 1\leq i\leq n.
\end{equation}
When $m=2$ and $g_i=\min_j f_{ij}$, (\ref{cond2}) is equivalent to 
\begin{equation}
\label{cond22}
0\leq \beta_1\leq \min_{i:\, f_{i1}>f_{i2}} \frac{f_{i2}}{f_{i1}-f_{i2}},
\quad 0\leq \beta_2\leq \min_{i:\, f_{i2}>f_{i1}} \frac{f_{i1}}{f_{i2}-f_{i1}}.
\end{equation}
In general, however, it is not clear how to reduce (\ref{cond2}) to an explicit range for $\beta$.  Although we do address the choice of $\beta$ in this section, further study is desired. 

Iteration (\ref{em1}) can also be derived from (\ref{llike00}), viewing it as a mixture log-likelihood with $2m+1$ components.  The density of the $j$th component ($0\leq j\leq 2m$) evaluated at observation $i$ is
 $$f^\#_{ij}=\begin{cases} g_i-\sum_{k=1}^m \tilde{f}_{ik} \beta_k, & j=0,\\
 \tilde{f}_{ij}, & 1\leq j\leq m, \\ \tilde{f}_{i,j-m}, & m+1\leq j\leq 2m.\end{cases}$$ 
The proportion of the $j$th component is $(2+\beta_+)^{-1}$ if $j=0$, $p_j/(2+\beta_+)$ if $1\leq j\leq m$, and $\beta_{j-m}/(2+\beta_+)$ if $m+1\leq j\leq 2m$.  The factor $2+\beta_+$ makes these proportions sum to one. 
Let $I^\#_{ij}$ denote the indicator of whether observation $i$ is from component $j$.  Similar to Section~2.2, if we treat $I^\#_{ij},\ 1\leq i\leq n,\ 0\leq j\leq 2m,$ as latent variables, then the resulting EM iteration is precisely (\ref{em1}). 

On the other hand, we can derive an EM iteration based on (\ref{llike2}), viewing it as a mixture log-likelihood with $m+1$ components, one of which has proportion $(2+\beta_+)^{-1}$, and the others have proportions $(p_j+\beta_j)/(2+\beta_+),\ j=1,\ldots, m$.  Equivalently, we treat 
\begin{equation}
\label{latent2}
I^\#_{i0},\ 1\leq i\leq n,\quad {\rm and}\quad I^\#_{ij}+I^\#_{i,j+m},\ 1\leq i\leq n,\ 1\leq j\leq m,
\end{equation}
as latent indicators instead of the entire collection $I^\#_{ij},\ 1\leq i\leq n,\ 0\leq j\leq 2m.$  The E-step, as before, is to calculate the conditional expectations.  We have 
$$E\left(I^\#_{ij}+\left. I^\#_{i,j+m}\right| y, \mathbf{p^{(t)}}\right) =\frac{\tilde{f}_{ij} \left(p_j^{(t)}+\beta_j\right)}{g_i +\sum_{k=1}^m \tilde{f}_{ik} p_k^{(t)}}\equiv K_{ij}.$$
The M-step seeks to maximize the function
\begin{equation}
\label{m2fun}
\sum_{j=1}^m \sum_{i=1}^n K_{ij} \log (p_j+\beta_j).
\end{equation}
By checking the Karush-Kuhn-Tucker conditions, it can be shown that (\ref{m2fun}) is maximized 
by choosing $p_j$ as
\begin{equation}
\label{em2}
p_j^{(t+1)}=\max\left\{0,\ \delta \sum_{i=1}^n K_{ij}-\beta_j\right\},\quad 1\leq j\leq m,
\end{equation}
where $\delta$ is determined by the constraint $\sum_{j=1}^m p_j^{(t+1)}=1$.  Iteration (\ref{em1}) corresponds to (\ref{em2}) with $\beta_j\equiv 0$. 

The new EM iteration (\ref{em2}) is more complicated than (\ref{em1}) or (\ref{em}), but only slightly so.  First, we observe that the right-hand side of (\ref{em2}) is a continuous and increasing function of $\delta$, and hence a $\delta$ exists to ensure $\sum_{j=1}^m p_j^{(t+1)}=1$, as long as $\sum_i K_{ij}>0$ for some $j$.  Moreover, such a $\delta$ can be found efficiently (in $O(m\log m)$ time) using a ``waterfilling'' algorithm (see Appendix~A in Yu 2010b). 

The convergence rate of (\ref{em2}) is no worse than that of (\ref{em1}), because the latent variables (\ref{latent2}) are less informative than the entire collection $I^\#_{ij},\ 1\leq i\leq n,\ 0\leq j\leq 2m,$ which lead to (\ref{em1}).  A slight extension of this argument reveals that, for fixed $g$, the convergence rate of (\ref{em2}) is monotonic in $\beta=(\beta_1, \ldots, \beta_m)$.  Such results resemble those of Yu (2010b) on improved Arimoto-Blahut algorithms for channel capacity calculations.  In Yu (2010b), the convergence rate comparison results are derived by calculating the matrix rate and analyzing its eigenvalues.  Because we work within the EM framework, however, the convergence rate comparison results here are obtained automatically once we specify the appropriate latent variables. 

Combined with the convergence rate comparisons of Section~2.2, the above considerations suggest the following guideline for choosing the squeezing parameters $g$ and $\beta$. 
\begin{itemize}
\item
Choose $g_i=\min_j f_{ij}$, which satisfies the upper bound in (\ref{cond1}).
\item
Choose $\beta_j$ to be as large as possible, subject to (\ref{cond2}). 
\end{itemize}
When $m=2$, the condition (\ref{cond2}) reduces to (\ref{cond22}).  Hence we recommend setting $\beta$ at the upper bounds in (\ref{cond22}). 

It is helpful to write down an explicit formula for (\ref{em2}) with these optimal choices of $g$ and $\beta$ in the $m=2$ case.  Actually, for later convenience, we present a slightly more general iterative formula for maximizing ($m=2$) 
\begin{equation}
\label{llike22}
\tilde{l}(\mathbf{p})=\sum_{i=1}^n \log\left(r_i+ f_{i1} p_1 +f_{i2} p_2\right),
\end{equation}
subject to $p_j\geq 0,\ j=1,2$, and $p_1+p_2=\beta_0$.  Here $r_i,\ i=1,\ldots, n,$ are nonnegative constants, and $\beta_0>0$ is fixed.  Define $g_i=\min\{f_{i1},\ f_{i2}\}$, and 
$$\beta_1=\min_{i:\, f_{i1}>f_{i2}} \frac{r_i+\beta_0 f_{i2}}{f_{i1}-f_{i2}},\quad \beta_2=\min_{i:\, f_{i2}>f_{i1}} \frac{r_i+\beta_0 f_{i1}}{f_{i2}-f_{i1}}.$$
Suppose the current parameter estimate is $\mathbf{p}^{(t)}$.  We compute 
$$S_j=\left(p^{(t)}_j +\beta_j\right)\sum_{i=1}^n \frac{f_{ij}-g_i}{r_i+f_{i1} p_1^{(t)}+f_{i2} p_2^{(t)}},\quad j=1,2.$$
Then we update $\mathbf{p}^{(t)}$ as
\begin{equation}
\label{em22}
p^{(t+1)}_j = \max\{0,\ \min\{\beta_0,\ (\beta_0 +\beta_1+\beta_2)S_j /(S_1+S_2)-\beta_j\}\},\quad j=1,2, 
\end{equation}
which is a slight generalization of (\ref{em2}) for $m=2$.  The iteration (\ref{em22}) is uniquely defined if $S_1+S_2\neq 0$, for which it suffices to have $f_{i1}-f_{i2}\not\equiv 0$.  (Inspection shows that $p_j^{(t)}+\beta_j>0,\ j=1,2,$ as long as $\tilde{l}(\mathbf{p}^{(t)})>-\infty$.)  If $f_{i1}-f_{i2}\equiv 0$, then the M-step is not unique, but it is convenient to set $\mathbf{p}^{(t+1)}=\mathbf{p}^{(t)}$.  Because of the $\max$ and $\min$ operations, (\ref{em22}) can potentially transfer all the mass from one component to the other in a single step.  This will be especially useful in later sections after we introduce nearest neighbor exchanges. 

\section{Nearest Neighbor Exchanges and the Cocktail Algorithm}
\subsection{Nearest Neighbor Exchanges}
The intuition that conventional EM tends to be slow when there exists heavy overlap between mixture components is used to our advantage in Section~2 for designing faster EM schemes via squeezing.  However, the strategies so far begin by squeezing out a common vector $g$ from {\it each} of the components.  When there exist many components, it is conceivable that squeezing applied to all components may not be effective, even though a sub-collection may have severe overlap.  Then it is worthwhile to apply some form of ``local squeezing'' to a sub-collection of components. 

For example, consider a nonparametric mixture problem where the observations $y_i$ are assumed to be drawn independently from a mixture of normals 
$$y_i\sim \sum_{j=1}^m p_j{\rm N}(\mu_j, 1).$$
The variance is fixed for simplicity.  We assume the mixing distribution puts mass $p_j$ on ${\rm N}(\mu_j, 1)$, where $\mu_j$ is obtained by discretizing an interval, say $\mu_j=Uj/m$, and $U>0$ denotes the largest of the normal means.  As $m$ increases, the overlap between adjacent densities ${\rm N}(\mu_{j}, 1)$ and ${\rm N}(\mu_{j+1}, 1)$ increases, and conventional EM slows down.  The global squeezing strategies may not be effective, however, because the overlap between the left-most density ${\rm N}(U/m, 1)$ and the right-most density ${\rm N}(U, 1)$ can still be small. 

A natural remedy, therefore, is to exchange the mass between each pair of nearby components in turn, holding the other components fixed.  This is similar to (but somewhat simpler than) the paired and rotated EM of Pilla and Lindsay (2001).  In general, given the current parameter estimate $\mathbf{p}^{(t)}$, let $j_1<\cdots<j_{q+1}$ be the elements of $\left\{j:\ p_j^{(t)}>0\right\}$ where $q+1$ is the number of support points of $\mathbf{p}^{(t)}$.  We perform mass exchanges between $j_k$ and $j_{k+1}$ for $k=1,\ldots, q$ in turn, i.e., 
\begin{equation}
\label{nne}
\mathbf{p}^{(t+k/q)}=VE\left(j_k, j_{k+1}, \mathbf{p}^{(t+(k-1)/q)}\right),\quad k=1,\ldots, q.
\end{equation}
We use $\mathbf{\tilde{p}}=VE(u, v, \mathbf{p}),\ u\neq v,$ to denote an update of the form 
$$\tilde{p}_j=\begin{cases} p_j, & j\notin \{u,\, v\},\\ p_j+\delta, & j=u,\\ p_j-\delta, & j=v,\end{cases}$$
where $\delta\in [-p_u, p_v]$ is chosen so that $l(\mathbf{\tilde{p}})\geq l(\mathbf{p})$.  To choose the step-length $\delta$, we naturally use (\ref{em22}).  The iteration (\ref{em22}) is applicable because, when other components are held fixed, the log-likelihood for $p_u, p_v$ is exactly in the form of (\ref{llike22}).  Because (\ref{em22}) is an EM iteration, the log-likelihood is automatically monotonic.  Moreover, (\ref{em22}) is easy to implement, and very amenable to theoretical analysis.  These all add to the appeal of (\ref{em22}) when compared with standard tools such as Newton's method.  We refer to the composite mapping $\mathbf{p}^{(t)}\to \mathbf{p}^{(t+1)}$ given by (\ref{nne}) as the set of nearest neighbor exchanges (NNEs). 

We have found that by adding nearest neighbor exchanges based on (\ref{em22}) to conventional EM can lead to considerably improved speed.  There is one caveat, however.  Conventional EM is usually started at the interior of the parameter space, i.e., $p_j^{(0)}>0$ for all $j$, because once a component receives zero mass, it does so in all subsequent iterations.  By adding nearest neighbor exchanges, certain $p_j$ may be set to zero.  While this has the desirable effect of eliminating bad support points, it may accidentally eliminate a good one, and yield a suboptimal solution.  The problem is easily remedied, however, by adding in the following step, known as the vertex direction method (VDM; Fedorov 1972).  Given the current parameter estimate $\mathbf{p}$, we first calculate the derivatives 
$$d_j=\frac{\partial l(\mathbf{p})}{\partial p_j}=\sum_{i=1}^n \frac{f_{ij}}{\eta_i},$$
where $\eta_i=\sum_{j=1}^m f_{ij} p_j$.  Let $j^\#$ denote any index such that $d_j,\ j=1,\ldots, m,$ is maximized.  Then we update $\mathbf{p}$ to $\mathbf{\tilde{p}}$ with 
\begin{equation}
\label{vdm}
\tilde{p}_j=\begin{cases} (1-\delta) p_j, & j\neq j^\#,\\ (1-\delta) p_j +\delta, & j=j^\#,\end{cases}
\end{equation}
where $\delta\in [0,1]$ is chosen so that $l(\mathbf{\tilde{p}})\geq l(\mathbf{p})$.  We use iteration (\ref{em22}) for choosing $\delta$ because $l(\mathbf{\tilde{p}})$ as a function of $(\delta, 1-\delta)$ is again in the form of (\ref{llike22}). 

Let us denote the mapping (\ref{vdm}) with $\delta$ chosen by (\ref{em22}) as $\mathbf{\tilde{p}}=VDM(j^\#, \mathbf{p})$.  In Section~3.2 we show that VDM based on (\ref{em22}), when added to conventional EM and nearest neighbor exchange iterations, results in a globally convergent algorithm.  That is, starting from any $\mathbf{p}^{(0)}\in \Theta$, all limit points of the resulting algorithm are global maxima of the log-likelihood function on $\Theta$.  The proof actually shows that by adding VDM to any monotonic algorithm we obtain a globally convergent algorithm.

\subsection{The Cocktail Algorithm}
We summarize a ``cocktail algorithm'' based on VDM, nearest neighbor exchanges, and EM steps.  A convergence proof is then provided.  Empirical evaluation of such a strategy is presented in Section~3.3.  Yu (2009) applies this strategy to the D-optimal design problem, and reports dramatic improvement in speed.  We show similar performance for the mixture problem. 

\begin{paragraph}
{\it Cocktail Algorithm}
\begin{description}
\item[1]
At iteration $t$, first perform a VDM step (\ref{vdm}) where $\delta$ is chosen using (\ref{em22}).
\item[2]
Then use the output of VDM, say $\mathbf{\tilde{p}}$, as input for the nearest neighbor exchanges, i.e., (\ref{nne}), again based on (\ref{em22}).  
\item[3]
Finally, update the output of (\ref{nne}) using (\ref{em}), i.e., conventional EM, to obtain the next iterate $\mathbf{p}^{(t+1)}$. 
\end{description}
\end{paragraph}

Note that this is only one of the potential algorithms based on reducing the overlap between component densities.  There is much room for further exploration.  One could consider, for example, an algorithm that uses only Steps~1 and 2 above at each iteration.  That is, a VDM step is combined with nearest neighbor exchanges.  We call this algorithm NNE+.  We design the cocktail algorithm in the hope that the nearest neighbor steps and conventional EM can complement each other, since NNE+ focuses on purely local modifications, whereas EM focuses on purely global ones.  As we shall illustrate in empirical examples, the performance of NNE+ or conventional EM (each by itself) can be poor, but the cocktail algorithm is very fast. 

Because the cocktail algorithm consists of many EM sub-steps, the log-likelihood is guaranteed to increase at each iteration.  Further analysis yields the following convergence theorem. 

\begin{theorem}
\label{conv}
The cocktail algorithm is globally convergent.  That is, if $\mathbf{p}^{(t)}$ is a sequence generated by the cocktail algorithm starting from any $\mathbf{p}^{(0)}\in \Theta$, then all limit points of $\mathbf{p}^{(t)}$ are global maxima of $l(\mathbf{p})$ on $\mathbf{p}\in \Theta$. 
\end{theorem}
\begin{proof}
The proof is similar to that of Theorem 1 in Yu (2009). 
Let $\mathbf{\tilde{p}}^{(t)}$ denote the output of the VDM step at iteration $t$.  By monotonicity, $$l(\mathbf{p}^{(t)})\leq l(\mathbf{\tilde{p}}^{(t)})\leq l(\mathbf{p}^{(t+1)}).$$
Hence the two sequences $l(\mathbf{p}^{(t)})$ and  $l(\mathbf{\tilde{p}}^{(t)})$ tend to the same (finite) limit.  Let $\mathbf{p}^*$ be a limit point of $\mathbf{p}^{(t)}$, and let $\mathbf{p}^{(t_j)}$ be a subsequence converging to $\mathbf{p}^*$.  
Without loss of generality, we may assume that the VDM steps $\mathbf{p}^{(t_j)}\to \mathbf{\tilde{p}}^{(t_j)}$ are all performed on the same index $k=j^\#$ as in (\ref{vdm}), since at least one of the $m$ indices will appear infinitely often.  If $f_{ik}\equiv \sum_j f_{ij} p_j$ for some $\mathbf{p}=\mathbf{p}^{(t_j)}$ or $\mathbf{p}=\mathbf{p}^*$, then we can show directly that $\partial l(\mathbf{p})/\partial p_k= n$.  By the choice of $k$, we have $\partial l(\mathbf{p})/\partial p_j\leq n$ for all $1\leq j\leq m$, and hence $\mathbf{p}$ is already a global maximum by the general equivalence theorem (Lindsay 1983).  Assume $f_{ik}\not\equiv \sum_j f_{ij} p_j$ for all $\mathbf{p}=\mathbf{p}^{(t_j)}$ and $\mathbf{p}=\mathbf{p}^*$.  Then inspection of (\ref{vdm}) and (\ref{em22}) reveals that all VDM steps $\mathbf{p}^{(t_j)}\to \mathbf{\tilde{p}}^{(t_j)}$ are uniquely defined.  Moreover, when $k$ is considered fixed, the VDM mapping is continuous at $\mathbf{p}^*$.  Hence $\mathbf{\tilde{p}}^{(t_j)}$ converges to $VDM(k, \mathbf{p}^*)=\mathbf{\tilde{p}}$, say, and $l(\mathbf{\tilde{p}})=l(\mathbf{p}^*)$ as a result.  Since this VDM step $\mathbf{p}^*\to \mathbf{\tilde{p}}$ is derived in the EM framework, $\mathbf{\tilde{p}}$ being uniquely defined means that it is the unique maximizer at the M-step.  If $\mathbf{\tilde{p}}\neq \mathbf{p}^*$, then the expected complete-data log-likelihood increases strictly, and so does the observed log-likelihood, which contradicts $l(\mathbf{\tilde{p}})=l(\mathbf{p}^*)$.  It follows that $\mathbf{\tilde{p}}= \mathbf{p}^*$, i.e., $\mathbf{p}^*$ is a fixed point of the VDM mapping.  Inspection of (\ref{vdm}) and (\ref{em22}), however, shows that this fixed point must satisfy $\partial l(\mathbf{p}^*)/\partial p_k\leq n$, which implies that $\mathbf{p}^*$ is a global maximum. 
\end{proof} 

\subsection{Numerical Illustration}
This section gives a numerical illustration of the effectiveness of the cocktail algorithm.  The cocktail algorithm is compared with conventional EM, i.e., iteration (\ref{em}), the algorithm NNE+ mentioned in Section~3.2, and the vertex exchange method (VEM) of B\"{o}hning (1985).  To describe VEM, suppose the parameter estimate at iteration $t$ is $\mathbf{p}^{(t)}$.  Define $d_j=\partial l\left(\mathbf{p^{(t)}}\right)/\partial p_j$, and let $j^\#$ and $j_\#$ be indices between 1 and $m$ such that 
$$d_{j_\#}=\min_{j:\, p_j^{(t)}>0} d_j,\quad d_{j^\#}=\max_{1\leq j\leq m} d_j.$$ 
VEM sets $\mathbf{p}^{(t+1)}$ as 
\begin{equation*}
\mathbf{p}^{(t+1)}=VE\left(j^\#, j_\#, \mathbf{p}^{(t)}\right).
\end{equation*}
That is, we exchange the mass between the indices $j^\#$ and $j_\#$ so as to increase the log-likelihood.  We again employ (\ref{em22}) to choose the step-length. 

We run each of (conventional) EM, NNE+, VEM, and the cocktail algorithm until convergence and record the number of iterations and computing time.  An iteration of NNE+ consists of one iteration of VDM and the set of nearest neighbor exchanges.  
An iteration of the cocktail algorithm consists of one iteration each of NNE+ and conventional EM.  We shall concentrate on the computing time as a more objective measure of performance.  All calculations are performed on the same Sun Solaris 10 machine, and the computing time is recorded using the R function system.time().  The program is written in C and is available, together with the R interface, upon request from the author. 

Each algorithm is started from the same uniform probability vector, i.e., $p^{(0)}_j=1/m,\ j=1,\ldots, m$.  We use the common convergence criterion 
\begin{equation}
\max_{1\leq j\leq m} d_j-n \leq \epsilon,
\label{epsilon}
\end{equation}
where $d_j=\partial l(\mathbf{p})/\partial p_j$.  The theoretical basis for (\ref{epsilon}) is that for any $\mathbf{p}$ that satisfies this criterion we have
$$l(\mathbf{\hat{p}})-l(\mathbf{p})\leq \epsilon,$$
where $\mathbf{\hat{p}}$ is the MLE (Lindsay 1983, B\"{o}hning et al.\ 1996).  We choose $\epsilon=10^{-6}$ in our 
experiments. 

EM, NNE+, VEM, and the cocktail algorithm are tested on data taken from Roeder (1990) concerning the velocities of 82 galaxies.  Following Pilla and Lindsay (2001), we fit a normal finite mixture model to these data.  The means of the normal components lie on a grid of 64 equally-spaced points from 10.0 to 33.94, and the common standard deviation is $\sigma=0.95$.  The algorithms deliver the same MLE as reported by Pilla and Lindsay (2001), and their performance is recorded in Table~1.  
%For this example, we also perform another experiment using code written entirely in R.  The comparison between %algorithms is similar.  As an indication of the ease of implementation of EM-type algorithms, this R code has fewer %than 80 lines. 

\begin{table}
\caption{Iteration count and computing time (in seconds) until convergence for four algorithms on the galaxy data.} 
\begin{center}
\begin{tabular}{rrrrrrrrrr}
\hline
    \multicolumn{4}{c}{Iteration count} &\multicolumn{4}{c}{Computing time}\\
     EM    & NNE+ & VEM  & Cocktail & EM   & NNE+ & VEM  & Cocktail\\
\hline
     21777 & 74   & 974  & 36       & 87   & 0.02 & 0.13 & 0.02\\
\hline
\end{tabular}
\end{center}
\end{table}

Clearly the nearest neighbor exchanges are very effective, since both NNE+ and the cocktail algorithm improve conventional EM dramatically, reducing its computing time by orders of magnitude.  Adding conventional EM to NNE+ appears to have increased the computing time per iteration, but decreased the number of iterations, so that the cocktail algorithm and NNE+ have similar overall computing time.  We remark that the cocktail algorithm is also appealing because it is easy to implement and requires virtually no tuning. 

\section{Efficient EM for Computing the NPMLE for Censored Data}
We show that the EM strategies designed in Sections~2 and 3, in particular the cocktail algorithm of Section~3, are also effective for the NPMLE problem for censored data.  Section~4.1 briefly reviews how this problem can be viewed as a problem of mixture proportions.  Section~4.2 highlights fast implementations of our algorithms.  Section~4.3 contains numerical illustrations using simulated data. 

\subsection{NPMLE for Censored Data}
Assume that failure time data collected from $n$ units are independent and identically distributed according to a  
distribution function $F$, except that they are subject to censoring.  Following Gentleman and Geyer (1994), 
assume there is an inspection time process $Q$ which is independent of the failure times, and suppose each unit is subject to inspections governed by $Q$ independently.  The observed data then consist of $n$ observation intervals $(l_i, r_i]$, where $l_i$ is the last inspection time prior to failure and $r_i$ is the first inspection time after failure for subject $i$.  Right censoring may be represented by $r_i=\infty$, and left censoring by $l_i=0$; the exact failure time is observed when $l_i$ coincides with $r_i$, i.e., when the individual is subject to continuous inspections. 

Doubly censored data arise when units are subject to both right and left censoring.  The inspection time 
process is given by a pair of random variables $(L, U)$, $L< U$.  If the failure time $T$ satisfies $L<T\leq U$, then $T$ is observed; if $T\leq L$, then the observation is left censored at $L$; if $T>U$, then the observation is right censored at $U$.  For (case 2) interval censored data, the failure time $T$ is not observed, but only known to fall within a random interval. 

In either case, let us order the distinct observation times (i.e., end points of the intervals $(l_i, r_i]$) as 
$0=z_0<z_1< \ldots <z_{m-1}<z_m=\infty$.  Denote 
$$\mathbf{p}=(p_1, \ldots, p_m),\quad p_j=F(z_j)-F(z_{j-1}),$$ 
and define the $n\times m$ matrix $(f_{ij})$ by  
$$f_{ij}=\begin{cases} 1, &  z_j\in (l_i, r_i],\\ 0, &\ {\rm otherwise}.\end{cases}$$
For notational convenience we assume that observation intervals are open at the left, and closed at the right, end points, though extension to the general case is trivial; if the $i$th observation is exact, for example, we simply set $f_{ij}=1$ if $z_j=l_i(=r_i)$ and $f_{ij}=0$ otherwise.  Seeking the maximizer $\hat{F}$ that jumps only at observed time points, we note that the log-likelihood function is exactly (\ref{llike}).  Thus all the EM-type algorithms in Sections~2 and 3 can be applied to compute the NPMLE, defined as the $\mathbf{\hat{p}}$ that maximizes (\ref{llike}). 

%Uniqueness of the maximiser in this univariate case is guaranteed and can be shown 
%by graph-theoretic methods (Gentleman and Vandal, 2001). 

\subsection{Efficient Implementations}
Straightforward implementation of conventional EM, NNE+, VEM, or the cocktail algorithm, as described in Sections~2 and 3, each requires $O(nm)$ time per iteration.  For univariate censored data, however, one can take advantage of the special structure of the matrix $(f_{ij})$ to derive $O(n)$ implementations.  This is a substantial reduction since $m$ is often of the same order of magnitude as $n$.  It is especially relevant in situations such as bootstrap resampling, when repeated use of the algorithms is needed. 

Noted only briefly by Jongbloed (1998), who does not give the technical details concerning computational complexity, this possibility of fast EM implementation has remained largely unnoticed.  Zhang and Jamshidian (2004) present fast implementations for doubly censored data, but not (case 2) interval censored data.  In Appendix~A, we show that all four algorithms admit efficient implementations for univariate interval censoring in general. 

\subsection{Evaluation of Algorithms on Doubly Censored Data} 
In this section we evaluate the effectiveness of EM, NNE+, VEM, and the cocktail algorithm for computing the NPMLE for doubly censored data.  Simulations are performed under conditions similar to those of Wellner and Zhan (1997), and Zhang and Jamshidian (2004).  Wellner and Zhan (1997) propose an effective algorithm which combines EM and iterative convex minorant (ICM; Jongbloed 1998) iterations.  We have decided to focus on evaluating the cocktail algorithm relative to EM, NNE+ and VEM because they are easy to describe and easy to implement.  A full evaluation, including the case of bivariate censoring, is work in progress. 

Our simulation setting is as follows.  The failure time $T_i$ for unit $i$ is generated as an independent 
exponential random variable with mean 1, and the censoring variables $L_i, U_i$ are generated as the $q_1$th and $q_2$th 
order statistics of 20 independent ${\rm uniform}(0,1)$ random variables.  The resulting observation is 
$$\begin{array}{rl}
T_i, & {\rm if\ } L_i<T_i\leq U_i;\\
(0, L_i], & {\rm if\ } T_i\leq L_i;\\
(U_i, \infty], & {\rm if\ } T_i> U_i.
\end{array}
$$
By adjusting $q_1$ and $q_2$ we obtain varying degrees of censoring. 

As in Section~3.3, all algorithms are started at the uniform probability vector, and the common convergence criterion is (\ref{epsilon}) with $\epsilon=10^{-6}$.  Our limited experience suggests that VEM may benefit from a starting value with fewer support points, but the cocktail algorithm is relatively insensitive to the initial number of support points.  Again, an iteration of the cocktail algorithm consists of an iteration of VDM, the set of nearest neighbor exchanges, and an iteration of conventional EM. 

Based on 10 replications, Tables~2 and 3 display the means and standard deviations of the number of iterations and computing 
time (in seconds) until convergence for EM, NNE+, VEM and the cocktail algorithm.  The input data are generated with $q_1=3$ and $q_2=18$ 
(moderate censoring) for Table~2, and with $q_1=8$ and $q_2=12$ (heavy censoring) for Table~3.

\begin{table}
\caption{Means and standard deviations of the number of iterations and computing time (in seconds) until convergence for four algorithms for doubly censored data.  Input data are generated with $q_1=3$ and $q_2=18$.}
\begin{center}
\begin{tabular}{lrrrrrrrr}
\hline
\        &\multicolumn{4}{c}{Iteration count} &\multicolumn{4}{c}{Computing time}\\
\        & EM    & NNE+   & VEM    & Cocktail & EM    & NNE+  & VEM    & Cocktail\\
\hline
\multicolumn{9}{c}{$n=1000$}\\
mean     & 5076  &  15456 & 5714   & 46.2     & 3.22  & 20.5  & 0.80   & 0.07    \\
s.d.     & 194   &  488   & 206    & 7.7      & 0.20  &  1.2  & 0.04   & 0.01    \\
\\
\multicolumn{9}{c}{$n=2000$}\\  
mean     & 9920 &  44044 & 12008  & 67.3     & 17.0  & 157    & 3.77   & 0.28 \\  
s.d.     & 892   &  1645  & 269    & 7.3      &  2.2  &   9   & 0.15   & 0.08 \\
\\
\multicolumn{9}{c}{$n=4000$}\\
mean     & 20347 & 100000+ & 24924 & 93.3     & 85.6   & 740+   & 19.1   & 0.77 \\ 
s.d.     &  1056  &        & 369    & 7.9     &  5.2   &        &  0.7   & 0.06 \\
\hline
\end{tabular}
\end{center}
\end{table}

\begin{table}
\caption{Means and standard deviations of the iteration count and computing time (in seconds) until convergence for four algorithms for doubly censored data.  Input data are generated with $q_1=8$ and $q_2=12$.} 
\begin{center}
\begin{tabular}{lrrrrrrrr}
\hline
\        &\multicolumn{4}{c}{Iteration count} &\multicolumn{4}{c}{Computing time}  \\
\        & EM     & NNE+  & VEM   & Cocktail & EM      & NNE+   & VEM    & Cocktail\\
\hline
\multicolumn{9}{c}{$n=1000$}\\
mean     & 5793   & 2768  & 2170  & 65.3     & 7.77    &  1.70  & 0.38   & 0.05    \\
s.d.     & 877    &  421  & 170   & 8.3      & 1.76    &  0.34  & 0.11   & 0.01    \\
\\
\multicolumn{9}{c}{$n=2000$}\\
mean     & 11034  &  8669 & 4411  & 103     & 38.0     &  13.0  & 1.56   & 0.20    \\
s.d.     &  2022  &  491  & 192   &  8      &  8.3     &  1.0   &  0.08  & 0.01    \\
\\
\multicolumn{9}{c}{$n=4000$}\\
mean     & 20397  & 27247 & 9176  & 145      & 163     &  89.8  & 8.76   & 0.61    \\
s.d.     & 5481   & 1957  & 336   & 17       & 51      &  8.8   & 0.57   & 0.08    \\
\hline
\end{tabular}
\end{center}
\end{table} 

In either situation we see that the cocktail algorithm is a dramatic improvement; it reduces the computing time of (conventional) EM or NNE+ by large factors, the reduction being more significant as $n$, the number of units, becomes larger.  The improvement is especially remarkable because the cocktail algorithm is a direct combination of EM and NNE+, each of which is very slow.  NNE+ performs much worse for moderately
censored data than for heavily censored data.  In the case of $q_1=3,\ q_2=18$ and $n=4000$, each of the 10 runs of NNE+ takes more than 100000 iterations.  As expected, each 
algorithm takes more iterations as $n$ increases.  Somewhat 
unexpectedly, for the same $n$, EM has similar iteration 
counts for moderately versus heavily censored data.  Yet the 
computing time of EM in the heavily censored case is significantly higher.  Another peculiarity is 
that, for the same $n$, VEM takes fewer iterations and less time for heavily censored data than 
for moderately censored data.  But the main feature in Tables~2 and 3 is the clear superiority of the 
cocktail algorithm in this example. 
%We are able to compute an accurate NPMLE for a data set with 
%4000 units comfortably under 1 second, thanks to the squeezing strategies of Section~2, the cocktail 
%algorithm of Section~3, and the efficient implementations of Appendix~A. 

\section{Discussion}
We have shown how to use efficient data augmentation to design fast EM-type algorithms for maximizing a mixture log-likelihood with known component densities.  Squeezing strategies are presented that take advantage of the overlap between components.  A cocktail algorithm that combines conventional EM with a nearest neighbor exchange 
strategy is found to perform very well for computing the NPMLE for censored data, which
is the intended application area of this work. 

The nearest neighbor exchange strategy works well with conventional EM when there is a natural ordering of the
mixture components, as in the case of univariate censored data.  It would be interesting to extend such algorithms to bivariate censoring (Betensky and Finkelstein 1999), where we observe a pair of possibly censored random variables for each unit.  Bivariate censoring presents many inferential and computational challenges.  Work on extending the effective nearest neighbor strategy is in progress, with encouraging preliminary results.  Extensions that accommodate truncation in addition to censoring, or that facilitate semi-parametric estimation, would also be desirable. 

It would be interesting to extend the squeezing strategies of Section~2 to mixture problems with 
unknown parameters in the component densities.  One approach is to again adopt AECM (Meng and van Dyk 1997),
and perform two types of EM-based maximization steps, one for the mixture proportions given the other parameters, and one for the other parameters given the mixture proportions.  The squeezing strategies can be used at the former maximization step.  It would be worthwhile to investigate the potential gain of using such strategies. 

\section*{Acknowledgment}
The author would like to thank Cliff Anderson-Bergman for helpful discussions on the squeezing strategies of 
Section~2. 

\appendix

\section*{Appendix A: Fast Implementations for the NPMLE Problem with Censored Data}
The conventional EM mapping, (\ref{em}), contains a summation of $n$ terms for each $j=1,\ldots, m$.  However, one can take advantage of the special structure of the matrix $(f_{ij})$ for fast computation.  This is a 
zero-one matrix whose non-zero entries are consecutive in each row.  Equivalently, if we let
\begin{equation}
a(i)=\min \{j:\ z_j\in (l_i, r_i]\},\quad b(i)=\max \{j:\ z_j\in (l_i, r_i]\},\quad i=1,\ldots, n,
\label{sp}
\end{equation}
then $f_{ij}=1$ if $a(i)\leq j\leq b(i)$ and $f_{ij}=0$ otherwise.  For convenience, we drop the 
superscripts in (\ref{em}) and focus on how to efficiently compute
$$p_j^{new}=\frac{1}{n}\sum_{i=1}^n \left(\frac{f_{ij}}{\sum_{k=1}^m f_{ik} p_k}\right) p_j,\quad j=1, \ldots, m,$$
for any $\mathbf{p}=(p_1, \ldots, p_m)\in \Theta$.  Note that computing $\eta_i=\sum_{k=1}^m f_{ik} p_k,\ 
i=1,\ldots, n,$ can be done in $O(m+n)$ time (or equivalently $O(n)$ time because $m\leq2n+1$) using Algorithm 1.

\begin{paragraph}
{\it Algorithm 1}
\begin{description}
\item[Step 1]  Calculate the cumulative sums $s_j=\sum_{k=1}^j p_k$, $j=1, \ldots, m$.  Set $s_0=0$.

\item[Step 2]  Set $\eta_i=s_{b(i)}-s_{a(i)-1}$, $i=1, \ldots, n$.
\end{description}
\end{paragraph}

If we can calculate
\begin{equation}
d_j=\sum_{i=1}^n f_{ij}/\eta_i=\sum_{i:\, a(i)\leq j\leq b(i)} 1/\eta_i,
\label{d} 
\end{equation}
then $p_j^{new}=d_jp_j/n$.  To calculate $d_j$ efficiently, we rely on the following algorithm.

\begin{paragraph}
{\it Algorithm 2}
\begin{description}
\item[Step 1]  Initialize $d_j=0,\ j=1, \ldots, m$.

\item[Step 2]  For $i=1, \ldots, n$, add $1/\eta_i$ to $d_{a(i)}$, and, if $b(i)+1\leq m$, subtract $1/\eta_i$ from
$d_{b(i)+1}$.

\item[Step 3]  Replace $(d_1, \ldots, d_m)$ by its cumulative sum, i.e., for $j=2, \ldots, m$, add $d_{j-1}$ to $d_j$.
\end{description}
\end{paragraph}

Obviously, Algorithm 2 costs $O(n)$ time.  We have

\begin{proposition}
Algorithm 2 is valid, i.e., its output agrees with (\ref{d}).
\label{valid1}
\end{proposition}
\begin{proof}
This can be shown by induction.  First, in the output of Algorithm 2, $d_1=\sum_{i:\, a(i)=1} 1/\eta_i$,
which agrees with (\ref{d}).  Assume Algorithm 2 gives the correct answer for $d_{j-1},\ j>1$.  Then Algorithm 2 computes
$d_j$ using $$d_j=d_{j-1}+\sum_{i:\, a(i)=j} 1/\eta_i-\sum_{i:\, b(i)+1=j} 1/\eta_i,$$
which again agrees with (\ref{d}) if we consider the difference $d_j-d_{j-1}$.  By the induction principle, all $d_j,\ j=1,
\ldots, m$, are correctly computed.
\end{proof}

Algorithms~1 \& 2 clearly give an $O(n)$ implementation of a conventional EM iteration.  Because $\eta_i,\ i=1,\ldots, n$ and $d_j,\
j=1,\ldots, m,$ are also the key quantities for VDM and VEM, the same efficient implementation applies to VDM and VEM. 
Specifically, the underlying iteration (\ref{em22}) is done in $O(n)$ time by keeping track of $\eta_i$. 

For NNE+ and the cocktail algorithm, we notice that each sub-step of nearest neighbor exchange in (\ref{nne}) affects only a limited 
number of terms in the log-likelihood.  Specifically, to implement sub-step $k$, only observation intervals $(l_i, r_i]$ 
such that either $a(i)\leq j_k\leq b(i)$ and $j_{k+1}> b(i)$, or $j_k<a(i)$ and $a(i)\leq j_{k+1}\leq b(i)$, need be 
considered.  Define the set 
$$V_k=\{i:\ a(i)\leq j_k\leq b(i),\ j_{k+1}> b(i)\}\cup \{i:\ j_k<a(i),\ a(i)\leq j_{k+1}\leq b(i)\}.$$ 
The time cost of sub-step $k$ is proportional to $|V_k|$, the number of entries in $V_k$.  However, because each $i$ belongs to at most two of $V_k,\ k=1,\ldots, q$, the total number of entries satisfy $\sum_{k=1}^{q} |V_k|\leq 2n$.  Hence, with a bit of bookkeeping, the entire set of nearest neighbor exchanges can be implemented in $O(n)$ time per iteration. 

{\bf Remark.} Algorithms 1 and 2 take $a(i),\ b(i),\ i=1,\ldots, n$, given by (\ref{sp}), as input.  Setting these up requires sorting 
the end points of the observation intervals $(l_i, r_i]$, which costs $O(n\log n)$ time.  Although slightly higher 
than the $O(n)$ per-iteration cost, this is typically a small fraction of the total computing time because of the required 
number of iterations.  Setting up the full matrix $(f_{ij})$, on the other hand, costs $O(mn)$ time, which is $O(n^2)$ in the worst 
case.

\end{document}